\documentclass[12pt]{article}
\usepackage{amssymb,amsmath,amsthm,latexsym}
\usepackage{
	linguex, 
	enumitem 
}
\usepackage{graphicx}
\usepackage{xcolor}

\usepackage[T1]{fontenc}

\usepackage{geometry}
\usepackage{cite}
\usepackage{amsmath}
\usepackage{amsthm}
\usepackage{amssymb,amsfonts,amsthm,latexsym}
\usepackage{algorithmic}
\usepackage{graphicx}
\usepackage{textcomp}
\usepackage{kotex}
\usepackage{
	linguex, 
	enumitem, 
	multirow
}
\usepackage{
	soul,
	mathrsfs,
	times,
	blkarray,
	comment
}
\usepackage{xcolor}

\theoremstyle{definition}
\newtheorem{theorem}{Theorem}

\newtheorem{corollary}[theorem]{Corollary}
\newtheorem{conjecture}[theorem]{Conjecture}

\newtheorem{algorithm}[theorem]{Algorithm}
\newtheorem{example}[theorem]{Example}
\newtheorem{definition}[theorem]{Definition}
\newtheorem{remark}[theorem]{Remark}

\def\x{{\mathbf x}}
\def\y{{\mathbf y}}

\def\0{{\mathbf 0}}
\def\1{{\mathbf 1}}
\def\r{{\mathbf r}}

\def\s{{\mathbf s}}

\def\F{{\mathbb F}}

\def\Z{{\mathbb Z}}

\def\R{{\mathcal R}}
\def\S{{\mathcal S}}

\def\I{{\mathbf I}}

\def\CC{{\mathcal C}}
\def\HH{{\mathcal H}}

\newenvironment{smatrix2}{\left(\begin{smallmatrix}}{\end{smallmatrix}\right)}
\newcommand{\dso}{d_{\mathrm{so}}}
\newenvironment{smatrix}[1]
{\arraycolsep=1pt\footnotesize
	\array{#1}}
{\endarray}

\usepackage[cmintegrals]{newtxmath}

\begin{document}

\title{Extending binary linear codes to self-orthogonal codes}

\author{
	Jon-Lark Kim\\
	Department of Mathematics\\
	Sogang University\\
	Seoul 04107, Korea\\
	e-mail: jlkim@sogang.ac.kr\\
	\\
Whan-Hyuk Choi\\
Department of Biomedical Engineering\\
UNIST\\
Ulsan 44919, Korea\\
e-mail: choiwh@unist.ac.kr\\
}


\maketitle

\begin{abstract}

Kim et al. (2021) gave a method to embed a given binary $[n,k]$ code $\mathcal{C}$ $(k = 3, 4)$ into a self-orthogonal code of the shortest length which has the same dimension $k$ and minimum distance $d' \ge d(\mathcal{C})$. We extend this result by proposing a new method related to a special matrix, called  the self-orthogonality matrix $SO_k$, obtained by shortening a Reed-Muller code $\mathcal R(2,k)$. Using this approach, we can extend binary linear codes to many optimal self-orthogonal codes of dimensions $5$ and $6$. Furthermore, we partially disprove the conjecture (Kim et al. (2021)) by showing that if $31 \le n \le 256$ and $n\equiv 14,22,29 \pmod{31}$, then there exist optimal $[n,5]$ codes which are self-orthogonal. We also construct optimal self-orthogonal $[n,6]$ codes when $41 \le n \le 256$ satisfies $n \ne 46, 54, 61$  and $n \nequiv 7, 14, 22, 29, 38, 45, 53, 60 \pmod{63}$.\\

\end{abstract}



\maketitle

\section{Introduction}

Since the beginning of the coding theory, many researchers have studied self-orthogonal (abbr. SO) codes and their applications.
It is well-known that they have connections to $t$-designs\cite{Bachoc2004} and lattices\cite{Harada2009}. Self-orthogonal codes also have connections to quantum codes\cite{Kim2002, Jin2010, Jin2012,Lisoneks2014}, which are currently receiving much attention due to quantum computers. 

One of the main topics in coding theory is to find the minimum distance optimal code among self-dual or self-orthogonal codes \cite{Niu2019}.
Boukllieve et al.\cite{Boukllieve2006} investigated optimal $[n,k]$ SO codes of lengths for $n \le 40$ and $k \le 10$. There are some optimal linear codes in BKLC (best-known linear codes) database of MAGMA\cite{Magma1994} which are also self-orthogonal. However, the highest minimum weight of some optimal $[n,k]$ SO codes are still unknown for some parameters $n$ and $k$.

Kim et al.\cite{Kim2021} gave a novel algorithm for the construction of optimal SO codes by adding columns to the generator matrix of a linear code $\CC$ of dimension $k \le 4$ and minimum distance $d' \ge d(\CC)$. They investigated the characterization of self-orthogonality for a given binary linear code in terms of the number of column vectors in its generator matrix. However, the algorithm in \cite{Kim2021} was not suitable to construct optimal SO codes of dimensions greater than or equal to 5.

In this paper, we extend this result for $k=5$ and $6$ by proposing a new method related to a special matrix, called  the self-orthogonality matrix $SO_k$, obtained by shortening a Reed-Muller code $\R(2,k)$.  Furthermore, we partially disprove the conjecture 25 in \cite{Kim2021} by showing that if $31 \le n \le 256$ and $n\equiv 14,22,29 \pmod{31}$, then there exist optimal $[n,5]$ codes that are self-orthogonal. We also show that if $n \le 256$ satisfies $n \ne 46, 54, 61, 86$  and $n \nequiv 7, 14, 22, 29, 38, 45, 53, 60 \pmod{63}$, then there exist optimal $[n,6]$ codes which are self-orthogonal. 

The paper consists of 5 sections. Section 2 gives preliminaries. Section 3 defines the self-orthogonality matrix $SO_k$ and describes our two main theorems, Theorems \ref{thm-SO-RM} and \ref{cor-SO embedding}. In Section 4, we propose a shortest SO embedding algorithm, partially disprove the conjecture 25 in \cite{Kim2021}, and construct optimal SO $[n,5]$ and $[n,6]$ codes. We also give an example of quantum codes based on self-orthogonal codes(see Corollary \ref{quantum_cor} and Example \ref{quantumcodeex}). We give a conclusion in Section 5.

\section{Preliminaries}

Let $\mathbb{F}$ be the finite field of order $2$.
A subspace $\CC$ of $\mathbb{F}^n$ is called a {\it linear code} of length $n$.
For $n,k \in \Z^+$, a $k$-dimensional linear code $\CC \subset \mathbb{F}^n$ is called an {\it $[n,k]$ code}.
The elements of $\CC$ are called {\it codewords}. A {\it generator matrix} for $\CC$ is a $k\times n$ matrix $G$ whose rows form a basis for $\CC$.

For $\x = (x_1,x_2,\ldots,x_n),\y = (y_1,y_2,\ldots,y_n) \in \mathbb{F}^n$, let $\x \cdot \y := \sum_{i=1}^n x_iy_i$.
For a linear code $\CC$, the code
$$\CC^{\perp} := \left\{ \x \in \mathbb{F}^n \; \middle| \; \x \cdot \y =0\: {\mbox{for  all }} \y \in \CC\right\}$$ is called the {\it dual} of $\CC$.
A linear code $\CC$ satisfying $\CC \subseteq \CC^\perp$ (resp. $\CC = \CC^\perp$) is called {\it self-orthogonal} (resp. {\it self-dual}).

For $\x,\y \in \mathbb{F}^n$, we define the {\it (Hamming) distance} $d(\x,\y)$ between $\x$ and $\y$ by the number of coordinates in which $\x$ and $\y$ differ. The {\it minimum distance} of $\CC$ is the smallest distance between any two distinct codewords.
For $n,k,d \in \Z^+$, an {\it $[n,k,d]$ code} $\CC$ is an $[n,k]$ code whose minimum distance is $d$.
A linear $[n,k]$ code $\CC$ is called $\it optimal$ if its minimum distance $d$ is the highest among all $[n,k]$ codes. For many values of $n$ and $k$, an optimal linear $[n, k]$ code is not self-orthogonal. So we say that a self-orthogonal $[n,k]$ code with the highest minimum weight among all self-orthogonal $[n, k]$ codes is an optimal SO code. We denote by $d(n,k)$ and $d_{\mathrm{so}}(n,k)$ the minimum distance of an optimal $[n,k]$ code and optimal $[n,k]$ SO code, respectively.

We point out that $d_{\mathrm{so}}(n,k)$ is always even because of self-orthogonality and that the best possible minimum distance of a self-orthogonal code is $2\lfloor d(n,k)/2 \rfloor$. In other words, if there exists a self-orthogonal code $\CC$ with the minimum distance $2\lfloor d(n,k)/2 \rfloor$, then $\CC$ is an optimal SO code and therefore, $d_{\mathrm{so}}(n,k)=2\lfloor d(n,k)/2 \rfloor.$

Let us collect some required notations. For any $[n,k]$ code $\CC$ generated by $G$, we denote by $\mathnormal{r}_i(G)$ the $i$th row of $G$ from the top for $1\le i \le k$  and $\mathnormal{c}_j(G)$ the $j$th column of $G$ from the left for $1 \le j \le n$. If there is no danger of confusion to the matrix $G$, then we will write $\mathnormal{r}_i$ (resp. $\mathnormal{c}_j$) for $\mathnormal{r}_i(G)$ (resp. $\mathnormal{c}_j(G)$). We denote by $H_k$ the generator matrix of the binary simplex code $\S_k$.
For example, $\S_3$ is the [7,3] linear code generated by $$H_3 =   \left(\begin{smatrix}{ccccccc}
0	& 0  &0&1&1&1&1\\
0	& 1 &1&0&0&1&1\\
1	& 0 &1&0&1&0&1
\end{smatrix}\right).$$

For $i=1,2, \cdots, 2^k-1$, we let
\begin{center} $\mathsf{h}_i :=$ the $i$th column vector of $H_k$, 
\end{center}
and for a $k\times n$ matrix $G$ over $\F$, we define
\begin{center} $\ell_{\mathsf{h}_i }(G) :=$ the number columns of $G$ which is equal to $\mathsf{h}_i $. \end{center}
If there is no confusion we will simply write $\ell_i$ for $\ell_{\mathsf{h}_i}$.
We also define a vector $\ell(G)$ over $\F$ as
\begin{center}  $\ell(G):=(\ell_1, \ell_2,\cdots ,\ell_{2^k-1}) \pmod 2.$\end{center}

\begin{example}\label{ex8_3}
	Let $\CC_{8,3}$ be a $[8,3,3]$ code generated by 	
	$$G_{8,3} =\left(\begin{smatrix}{cccccccc}
	1&1&0&1&1&1&1&0\\
	0&0&1&1&1&0&0&1\\
	0&0&0&0&0&1&1&1
	\end{smatrix}\right).$$
	Then $c_1=\mathsf{h}_4, c_2=\mathsf{h}_4,c_3=\mathsf{h}_2,c_4=\mathsf{h}_6,c_5=\mathsf{h}_6,c_6=\mathsf{h}_5, c_7=\mathsf{h}_5,c_8=\mathsf{h}_3$, thus $$(\ell_1, \ell_2,\ell_{3},\ell_{4},\ell_{5},\ell_{6},\ell_{7})=( 0, 1, 1, 2, 2, 2, 0),$$
	and we obtain the binary vector
	$$\ell(G_{8,3})=( 0, 1,1, 0, 0, 0, 0).$$
\end{example}

\section{Binary self-orthogonality Matrix}

\begin{definition}
	For the matrix $H_k$, let $\r_i$ be the $i$th row vector of $H_k$. Following the notation of \cite[Thm 1.4.3.(i)]{HP2010}, let $\r_i \cap \r_j$ be the vector in $\F^n$, which has 1s precisely in those positions where both $\r_i$ and $\r_j$ have 1s. Let $R(H_k)$ be a set of vectors as $$R(H_k)=\{ \r_i \cap \r_j \mid 1\le i \le j \le k \}.$$ Then, we define {\it the self-orthogonality matrix} $SO_k$ as a matrix with all vectors in $R(H_k)$ as rows. Since $R(H_k)$ is a set of $\frac{k(k+1)}{2}$ vectors of length $2^k-1$, the size of $SO_k$ is $\frac{k(k+1)}{2} \times (2^k-1)$.
\end{definition}

We note that the vector $\r_i \cap \r_i$ is equal to $\r_i$. 
Therefore, all the rows of $H_k$ are also rows of the matrix $SO_k$, thus we regard $SO_k$ as a vertically concatenated matrix of $H_k$ and the matrix consisting of the row vectors $\r_i \cap \r_j$ for $1\le i < j \le k.$

\begin{example}
	Since $H_2=\left(\begin{smatrix}{ccc}
	0&1&1\\
	1&0&1
	\end{smatrix}\right)$, the set $R(H_2)$ has three vectors $\r_1 \cap \r_1$, $\r_2 \cap \r_2$, and $\r_1 \cap \r_2$.
	Thus, 
	$$SO_2=\left(\begin{smatrix}{c}
	\r_1 \cap \r_1 \\
	\r_2 \cap \r_2\\ \hline
	\r_1 \cap \r_2
	\end{smatrix}\right)=\left(\begin{smatrix}{ccc}
	0&1&1\\
	1&0&1\\ \hline
	0&0&1
	\end{smatrix}\right)$$

\end{example}

\begin{example}
	Since $H_3 =   \left(\begin{smatrix}{ccccccc}
	0	& 0  &0&1&1&1&1\\
	0	& 1 &1&0&0&1&1\\
	1	& 0 &1&0&1&0&1
	\end{smatrix}\right)$, the set $R(H_3)$ has 6 vectors and
	$$SO_3 =\left(\begin{smatrix}{c}
	\r_1 \cap \r_1\\
	\r_2 \cap \r_2\\
	\r_3 \cap \r_3\\ \hline
	\r_1 \cap \r_2\\
	\r_1 \cap \r_3\\
	\r_2 \cap \r_3
	\end{smatrix}\right)=\left(\begin{smatrix}{ccccccc}
	0	& 0  &0&1&1&1&1\\
	0	& 1 &1&0&0&1&1\\
	1	& 0 &1&0&1&0&1\\ \hline
	0& 0 &0&0&0&1&1	\\
	0& 0 &0&0&1&0&1\\
	0& 0 &1&0&0&0&1
	\end{smatrix}\right).$$
\end{example}

\begin{definition}[\cite{Kim2021}]
	Let $\CC$ be an $[n,k]$ code generated by $G$.
	
	\begin{enumerate}[label = (\arabic*)]
		\item An \emph{SO embedding} of $\CC$ is an SO code whose generator matrix $\tilde{G}$ is obtained by adding a set $S$ of column vectors to $G$, i.e.,
		$\tilde{G} := [G~ | ~S]$.
		\item An SO embedding of $\CC$ is called a \emph{shortest SO embedding} of $\CC$ if its length is shortest among all SO embeddings of $\CC$.
	\end{enumerate}
	
\end{definition}

Kim et al.~\cite[p. 3705]{Kim2021} remarked the following based on their complicated algorithms and left the other cases open, which will be solved in Theorem \ref{cor-SO embedding}.

\begin{remark}\label{remark7}
	
	\begin{enumerate}[label = (\arabic*)]
		
		\item A shortest SO embedding code of a binary $[n, 2]$ code and a binary $[n, 3]$ code can be obtained by adding exactly three or fewer columns.
		
		\item A shortest SO embedding code of a binary $[n, 4]$ code can be obtained by adding exactly five or fewer columns.
	\end{enumerate}
	
\end{remark}

For a $k \times n$ matrix $G$ and $0<j\le k$, let $\I(j)$ be a multiset
\begin{displaymath}
\I(j):= \left\{ c_i(G) \left|  \begin{array}{l}
(i)~ 1\le i \le n \\
(ii)~ c_i(G) = \mathsf{h}_t \text{ for } 1 \le t \le 2^k-1 \\\text{ satisfying } \lfloor{\frac{t}{2^{j-1}}}\rfloor \equiv_2 1
\end{array} \right. \right\}. \end{displaymath}
In other words, $\I(j)$ is a multiset of the columns which has a 1 in the $j$-th row from the bottom of $G$.

Then the next theorem is proved as a characterization for self-orthogonality in terms of $\I(j)$ in \cite{Kim2021}.
\begin{theorem}[Lemma 2 \cite{Kim2021}]\label{cond} Let $\CC$ be an $[n,k]$ code generated by $G$. Then, $\CC$ is self-orthogonal if and only if for all $1\le j \le j' \le k$, $|\I(j)\cap \I(j')|$ is even.
\end{theorem}

The following main Theorem characterizes self-orthogonality using the vector $\ell(G)$ and the matrix $SO_k$.
\begin{theorem}\label{so_check}
	Let $\CC$ be a binary $[n,k]$ code generated by $G$. Then, $\CC$ is self-orthogonal if and only if $$SO_k\ell(G)^T = \0.$$
\end{theorem}
\begin{proof}
	By Theorem \ref{cond}, we know that $\CC$ is self-orthogonal if and only if for all $1\le j \le j' \le k$,	$|\I(j)\cap \I(j')|$ is even. From definitions of $SO_k$, $\ell(G)$, and $\I(j)$, it is easy to check that $|\I(j)\cap \I(j')| = (\r_{k+1-j} \cap \r_{k+1-j'})\cdot \ell(G)$ for all $1\le j \le j' \le k$. Therefore, $\CC$ is self-orthogonal if and only if for all $1\le j \le j' \le k$,	$(\r_{k+1-j} \cap \r_{k+1-j'})\cdot \ell(G)= 0$, equivalently, $SO_k \ell(G)^T = 0.$
\end{proof}	

\begin{corollary}\label{lem_so_check}
	Let $\CC$ be a binary $[n,k]$ code generated by $G$ and let $\HH$ be a binary code generated by the matrix $SO_k$.  Then, $\CC$ is self-orthogonal if and only if the vector $\ell(G)$ belongs to $\HH^\perp$.
\end{corollary}
\begin{proof}
	By Theorem \ref{so_check}, we know that $\CC$ is self-orthogonal if and only if $SO_k\ell(G)^T = \0.$ Since $SO_k$ is a generator matrix of $\HH$, the corollary follows.
\end{proof}

\begin{example}\label{ex_11}~
	\begin{enumerate}[label = (\arabic*)]
		\item Let $G_{8,3}$ be the generator matrix of $\CC_{8,3}$ in Example \ref{ex8_3}. We have 
		$$SO_3 =\left(\begin{smatrix}{ccccccc}
		0	& 0  &0&1&1&1&1\\
		0	& 1 &1&0&0&1&1\\
		1	& 0 &1&0&1&0&1\\ 
		0& 0 &0&0&0&1&1	\\
		0& 0 &0&0&1&0&1\\
		0& 0 &1&0&0&0&1
		\end{smatrix}\right), \text{ and } 
		\ell(G_{8,3})=0 11 0 0 0 0$$
		Thus, 
		$$SO_3 \ell(G_{8,3})^T = \left(\begin{smatrix}{ccccccc}
		0	& 0  &0&1&1&1&1\\
		0	& 1 &1&0&0&1&1\\
		1	& 0 &1&0&1&0&1\\ 
		0& 0 &0&0&0&1&1	\\
		0& 0 &0&0&1&0&1\\
		0& 0 &1&0&0&0&1
		\end{smatrix}\right)\left(\begin{smatrix}{c}
		0\\
		1	\\
		1\\
		0\\ 
		0	\\
		0\\
		0
		\end{smatrix}\right)= \left(\begin{smatrix}{c}
		0\\
		0	\\
		1\\
		0\\ 
		0	\\
		1
		\end{smatrix}\right),$$ and $\CC_{8,3}$ is not self-orthogonal by Theorem \ref{so_check}.
		
		\item  Let ${G_{10,3}}$ be a matrix
		$${G_{10,3}}=\left(\begin{smatrix}{cccccccccc}
		1&1&0&1&1&1&1&0&0&0\\
		0&0&1&1&1&0&0&1&1&1\\
		0&0&0&0&0&1&1&1&0&1
		\end{smatrix}\right).$$ Then $
		\ell({G_{10,3}})=0000000,$ and we have 
		$$SO_3 \ell(G_{10,3})^T = \0.$$ Therefore, by Theorem \ref{so_check}, we know that ${{G}_{10,3}}$ generates a self-orthogonal code.

	\end{enumerate}
	
\end{example}

In the next theorem, we discuss the relationship between the self-orthogonality matrix $SO_k$ and Reed-Muller code  $\R(r,m)$ defined in \cite[Chapter 1.10]{HP2010}.

\begin{theorem}\label{thm-SO-RM}
	Let $\R(r,m)$ be the $r$th order binary Reed-Muller code of length $2^m$. Then for $k \ge 3$,
	
	\begin{enumerate}[label = \roman*)]
		\item $SO_k$ is obtained by shortening on the first coordinate position from a binary Reed-Muller code $\R(2, k)$. Hence $\ell(G)$ is orthogonal to  $SO_k$ if and only if  $\ell(G)$ belongs to  the puncture of a Reed-Muller code $\R(k-3, k)$ on the first position.
		\item The covering radius $\rho$ of the puncture of a Reed-Muller code $\R(k-3, k)$ is $k+1$ if $k$ is even and $k$ if $k$ is odd.
		
	\end{enumerate}
\end{theorem}
\begin{proof}
	i) From the definition of $SO_k$, it is easy to see that
	$SO_k$ is obtained by shortening on the first coordinate position from a binary Reed-Muller code $\R(2, k)$. The dual of this shortened RM code is the puncture of dual of $\R(2, k)$ on the first position by \cite[Theorem 1.5.7]{HP2010}, and so it is equal to the puncture of $\R(k-3, k)$.
	
	ii)
	By \cite{Aileen1979}  the covering radius  $\rho$ of a Reed-Muller code $\R(k-3, k)$ is
	$k+2$ if $k$ is even and
	$k+1$ if $k$ is odd.
	It is easy to see that $\R(k-3, k)$ is even since all-one vector is in its dual. Therefore by \cite[Ex. 625]{HP2010},  the covering radius $\rho$ of the puncture of a Reed-Muller code $\R(k-3, k)$ is $k+1$ if $k$ is even and $k$ if $k$ is odd.
\end{proof}

\begin{corollary}
	The rank of the matrix $SO_k$ is $\frac{k(k+1)}{2}$.
\end{corollary}
\begin{proof}
	By \cite[Theorem 1.5.7]{HP2010}, we know that the dimension of Reed-Muller code $\R(2, k)$ equals $$ {k \choose 0} + {k \choose 1} + {k \choose 2}=1+k+\frac{k(k-1)}{2}.$$ Since $SO_k$ is obtained by shortening on the first coordinate position of $\R(2, k)$, the rank of the matrix $SO_k$ is one less than  the dimension of $\R(2, k)$. Thus the rank of the matrix $SO_k$ is $\frac{k(k+1)}{2}.$
\end{proof}

The following theorem generalizes Remark \ref{remark7} for any $k \ge 2$.
\begin{theorem}\label{cor-SO embedding}
	Given any binary $[n,k]$ code $\CC$ with $k \ge 2$, we can obtain a shortest SO embedding $\tilde{\CC}$ by adding exactly or less 
	
	\begin{enumerate}[label = \roman*)]
		\item $k+1$ columns if $k$ is even or
		\item $k$ columns if $k$ is odd.
	\end{enumerate}
\end{theorem}

\begin{proof}
	By Remark \ref{remark7}.(1), we may assume $k \ge 3$. Let $\CC$ be a given binary $[n, k]$ code with generator matrix $G$. By Theorem~\ref{thm-SO-RM}, we can check whether it is self-orthogonal. If not, by ii) of Theorem~\ref{thm-SO-RM} any binary vector $\ell(G)$ can be corrected by changing at most $k+1$ positions of ${\bf v}$ if $k$ is even and $k$ positions of ${\bf v}$ if $k$ is odd. Since we are embedding $\CC$ into a self-orthogonal codes, we need to add $k+1$ columns (corresponding to those $k+1$ positions) to $G$ if $k$ is even, and $k$ columns (corresponding to those $k$ positions) to $G$ if $k$ is odd.
\end{proof}

\section{Shortest SO embedding algorithm}
In what follows, we describe a {\it shortest SO embedding algorithm} to construct an embedding of a given binary $[n, k]$ code. Normally, we consider a best-known linear $[n, k]$ code from $BKLC$ database in Magma~\cite{Magma1994}, hoping to obtain an optimal or best-known self-orthogonal code with the same dimension. 

\begin{algorithm}[Shortest SO embedding algorithm]\label{alg1}\hfill
	
	{\fontsize{10}{10} \sf
		\noindent $\bullet$ Input: A generator matrix $G$ of the $[n,k]$ code $\CC$. 
		
		\noindent $\bullet$ Output: A generator matrix $G'$ for a shortest SO embedding.
		\begin{enumerate}[label =  ({A}\arabic*)] 
			\item Obtain the vector $\ell(G)$ from $G$. Go to ({\bf A2}).
			\item If $SO_k \ell(G)^T = \0$, then $\CC$ is self-orthogonal. Go to ({\bf A6}).
			\item  If $SO_k \ell(G)^T \ne \0$, then let $\s=SO_k \ell(G)^T $ be the syndrome corresponding to the vector $\ell(G)$, and go to ({\bf A4}).
			\item Let $err_{\s}$ be a coset leader correspond to the syndrome $\s$ and go to ({\bf A5}).
			\item For the each index $i$ of non-zero elements in $err_{\s}$, append $G$ the column vector which is the binary representation of $i$. Go to ({\bf A6}). 
			\item Return $G$ and terminate the algorithm.
		\end{enumerate}
	}
\end{algorithm}

We remark that a discussion of various decoding algorithms for Reed-Muller codes can be found in \cite{abbe2021}.

\begin{example}~
	\begin{enumerate}[label = (\arabic*)]
		\item Let $G_{8,3}$ be the generator matrix of $\CC_{8,3}$ in Example \ref{ex8_3}. In Example \ref{ex_11}, we obtained the syndrome $\s = 1001000^T$ of the vector $\ell(G_{8,3})=0110000.$ In this case, $\ell(G_{8,3})$ itself is a coset leader $err_{\s}$ corresponding to the syndrome $\s$. Therefore, by adding two columns that are binary representations of 2 and 3, we obtain a matrix $${\widetilde{G}_{8,3}}=\left(\begin{smatrix}{cccccccc||cc}
		1&1&0&1&1&1&1&0&0&0\\
		0&0&1&1&1&0&0&1&1&1\\
		0&0&0&0&0&1&1&1&0&1
		\end{smatrix}\right),$$ which generates a self-orthogonal $[10, 3, 4]$ code.		
		
		\item Let $\CC_{11,4}$ be the best-known linear $[11, 4,5]$ code from $BKLC$ database in Magma~\cite{Magma1994} with generator matrix 
		$$G_{11,4}=\left(\begin{smatrix}{cccccccccccc}
		1 &0& 0& 1& 0& 1& 0 &1& 1& 0& 1\\
		0 &1& 0& 1&1& 0& 0 &1 &0 &1& 1\\
		0& 0& 1& 1& 1& 1& 0& 0& 1 &1& 1\\
		0& 0& 0& 0& 0& 0& 1& 1& 1& 1& 1\end{smatrix}\right).$$
		Then 
		$\ell(G_{11,4})=110101110110111$.
		
		Since $$SO_4 = \left(\begin{smatrix}{cccccccccccccccc}
		0 &0&0&0& 0 &0&0 &1& 1 &1&1&1& 1 &1&1\\
		0 &0&0&1& 1 &1&1 &0& 0 &0&0&1& 1 &1&1\\
		0 &1&1&0& 0 &1&1 &0& 0 &1&1&0& 0 &1&1\\
		1 &0&1&0& 1&0&1  &0& 1 &0&1&0& 1&0&1\\ \hline
		0 &0&0&0& 0 &0&0 &0& 0 &0&0&1& 1 &1&1\\
		0 &0&0&0& 0 &0&0 &0& 0 &1&1&0& 0 &1&1\\
		0 &0&0&0& 0&0&0 &0& 1 &0&1&0& 1&0&1\\ 
		0 &0&0&0& 0 &1&1 &0& 0 &0&0&0& 0 &1&1\\
		0&0&0&0& 1&0&1  &0& 0 &0&0&0& 1&0&1\\
		0 &0&1&0& 0&0&1  &0& 0 &0&1&0& 0&0&1
		\end{smatrix}\right),$$ its syndrome is $$\s= SO_4\ell(G_{11,4})^T=0011101011^T.$$
		Since the coset leader corresponding to this syndrome $\s$ is $000001000100001$, we add three columns which are the binary representations of 6,10, and 15, respectively. 
		Consequently, we obtain a $[14,4,6]$ self-orthogonal code with generator matrix
		$$\widetilde{G}_{11,4}=\left(\begin{smatrix}{cccccccccccccccc}
		1&0&0&1&0&1&0&1&1&0&1&0&1&1\\
		0&1&0&1&1&0&0&1&0&1&1&1&0&1\\
		0&0&1&1&1&1&0&0&1&1&1&1&1&1\\
		0&0&0&0&0&0&1&1&1&1&1&0&0&1\end{smatrix}\right).$$

	\end{enumerate}
	
\end{example}

\subsection{Optimal self-orthogonal codes of dimension 5}
In \cite{Boukllieve2006}, Bouyuklliev et al. determined parameters of optimal $[n,5]$ SO codes of lengths up to 40. However, using Algorithm \ref{alg1}, we successfully constructed many new optimal $[n,5]$ SO codes of lengths up to 256 starting from best-known linear codes($BKLC$ database) of MAGMA\cite{Magma1994} and their punctured codes on at most six columns.

In \cite{Kim2021}, authors gave the next conjecture based on their computational results.

\begin{conjecture}[Conjecture 25 in \cite{Kim2021}]\label{conj}
	For $n\ge 10$, if $n \neq 13$ and $n \equiv 6,13,14,21,22,28,29 \pmod{31}$, then $\dso(n,5) = d(n,5) - 2$, i.e., there are no $[n,5,d(n,5)]$ SO codes.
\end{conjecture}

However, the conjecture is partially disproved. We disproved the conjecture partially by constructing new $[n,5,d(n,5)]$ SO codes for $n \ge 31$ and $n\equiv 14,22,29 \pmod{31}$ using Algorithm \ref{alg1}. We give three of these counterexamples in Example \ref{counterex} and the other counterexamples are presented in the web \cite{Choi2021} due to lack of space. The new parameters of entire counterexamples are listed in Table \ref{table:dso5}.

\begin{theorem}\label{thm17}
	There exist $[n,5,d(n,5)]$ SO codes for $32 \le n \le 256$ and $n\equiv 14,22,29 \pmod{31}$. 
\end{theorem}

\begin{example}\label{counterex}
	We give three counterexamples of Conjecture 19 in \cite{Kim2021} as follows.
	\begin{enumerate}[label = (\arabic*)]
		\item Adding three columns on a $[42, 5, 19]$ code,
		we obtain an optimal  $[45, 5, 22]$ SO code with generator matrix
		
		$\begin{smatrix2}
		1 0 0 1 0 1 0 1 1 0 1 0 0 1 0 0 1 0 1 1 1 1 0 1 0 1 0 0 1 0 1 1 0 0 1 0 1 0 1 1
		0 1 0 1 1\\
		0 1 0 1 1 0 0 1 0 1 1 0 0 0 1 0 0 1 1 0 1 0 1 1 0 0 1 0 1 1 0 0 1 0 1 1 0 0 1 0
		1 1 1 1 1\\
		0 0 1 1 1 1 0 0 1 1 1 1 0 1 1 1 0 0 1 1 1 0 0 0 0 1 1 0 0 0 0 1 1 0 0 0 0 1 1 0
		0 0 1 0 1\\
		0 0 0 0 0 0 1 1 1 1 1 1 0 0 0 0 1 1 1 1 1 1 1 1 1 0 0 0 0 0 0 1 1 1 1 1 1 0 0 0
		0 0 1 0 0\\
		0 0 0 0 0 0 0 0 0 0 0 0 1 1 1 1 1 1 1 1 1 1 1 1 1 1 1 1 1 1 1 0 0 0 0 0 0 0 0 0
		0 0 1 1 1
		\end{smatrix2}$.\\
		
		\item Adding three columns on a $[50, 5, 23]$ code,
		we obtain an optimal  $[53, 5, 26]$ SO code with generator matrix
		
		$\begin{smatrix2}
		1 0 0 1 0 1 1 0 0 1 1 0 1 0 0 1 1 1 0 1 1 0 1 1 1 1 0 1 0 0 1 1 0 0 1 0 1 1 0 0
		1 0 1 1 0 0 1 1 0 1 0 0 0\\
		0 1 0 1 0 1 0 1 0 1 0 1 0 1 0 1 0 0 1 1 0 1 1 1 0 1 0 1 0 1 0 1 0 1 0 1 0 1 0 1
		0 1 0 1 0 1 0 1 0 1 0 0 0\\
		0 0 1 1 0 0 1 1 1 1 0 0 1 1 0 0 1 0 0 1 1 1 0 1 1 0 0 1 1 0 0 0 0 1 1 0 0 1 1 0
		0 1 1 0 0 0 0 1 1 0 0 1 1\\
		0 0 0 0 1 1 1 1 1 1 1 1 0 0 0 0 0 1 1 1 1 1 1 0 0 0 0 1 1 1 1 1 1 1 1 0 0 0 0 0
		0 1 1 1 1 1 1 1 1 0 1 0 1\\
		0 0 0 0 0 0 0 0 0 0 0 0 0 0 1 1 1 1 1 1 1 1 1 1 1 1 1 1 1 1 1 1 1 1 1 1 1 0 0 0
		0 0 0 0 0 0 0 0 0 0 1 1 1
		\end{smatrix2}$.\\

		\item Adding two columns on a $[58, 5, 29]$ code,	
		we obtain an optimal  $[60, 5, 30]$ SO code with generator matrix
		
		$\begin{smatrix2}
		1 0 0 1 0 1 1 0 0 1 1 0 1 0 0 1 1 0 0 1 0 1 1 0 0 1 1 0 1 0 0 1 1 0 1 0 0 1 1 0
		0 1 0 1 1 0 0 1 1 0 1 0 0 1 1 0 0 1 0 1\\
		0 1 0 1 0 1 0 1 0 1 0 1 0 1 0 1 0 1 0 1 0 1 0 1 0 1 0 1 0 1 0 1 0 1 0 1 0 1 0 1
		0 1 0 1 0 1 0 1 0 1 0 1 0 1 0 1 0 1 1 0\\
		0 0 1 1 0 0 1 1 0 0 1 1 0 0 1 1 1 1 0 0 1 1 0 0 1 1 0 0 1 1 1 1 0 0 1 1 0 0 1 1
		0 0 1 1 0 0 0 0 1 1 0 0 1 1 0 0 1 1 1 1\\
		0 0 0 0 1 1 1 1 0 0 0 0 1 1 1 1 1 1 1 1 0 0 0 0 1 1 1 1 0 0 0 0 0 0 1 1 1 1 0 0
		0 0 1 1 1 1 1 1 1 1 0 0 0 0 1 1 1 1 0 0\\
		0 0 0 0 0 0 0 0 1 1 1 1 1 1 1 1 1 1 1 1 1 1 1 1 0 0 0 0 0 0 0 0 0 0 0 0 0 0 1 1
		1 1 1 1 1 1 1 1 1 1 1 1 1 1 0 0 0 0 0 0
		\end{smatrix2}$.

	\end{enumerate}	
\end{example}

Based on Theorem \ref{thm17}, we give modified conjectures as follows.

\begin{conjecture}\label{new_conj1}
	For $n\ge 32$, if $n\equiv 14,22,29 \pmod{31}$, then $\dso(n,5) = d(n,5)$, i.e., there exist $[n,5,d(n,5)]$ SO codes.
\end{conjecture}

\begin{conjecture}\label{new_conj2}
	If $n=14, 21,22,28,29$ or if $n \ge 32$ and $n\equiv 6, 13,21,28 \pmod{31}$, then $\dso(n,5) = d(n,5) - 2$, i.e., there are no $[n,5,d(n,5)]$ SO codes.
\end{conjecture}

\begin{table}
	\begin{center}
		\begin{tabular}{c|c|c|c}
			\multirow{2}{*}{$n$}&	\multirow{2}{*}{$\pmod{31}$}	&\multicolumn{2}{c}{$\dso(n,5)$}\\\cline{3-4}
			
			&& Known[ref] 	& Our result\\\hline\hline
			$45$& 14&	$\le$ 22\cite{Griesmer1960} 		    &${22}$	\\\hline
			$53$&22&	$\le$26\cite{Griesmer1960}		    &	${26}$	\\\hline
			$60$&29&$\le$30\cite{Griesmer1960}		&	${30}$	\\\hline
			$76$&14&$\le$38\cite{Griesmer1960}	    &	${38}$	\\\hline
			$84$&22&$\le$42\cite{Griesmer1960}	&	${42}$		\\\hline
			$91$&29&$\le$46\cite{Griesmer1960}	   &	${46}$		\\\hline
			$107$&14&$\le$54\cite{Griesmer1960}	   &	${54}$		\\\hline
			$115$&22&$\le$58\cite{Griesmer1960}	    &${58}$	\\\hline
			$122$&29&$\le$62\cite{Griesmer1960}      &	${62}$			\\\hline
			$138$&14&$\le$70\cite{Griesmer1960}	    &${70}$			\\\hline
			$146$&22&	$\le$74\cite{Griesmer1960}    &${74}$     \\\hline
			$153$&29&	$\le$78\cite{Griesmer1960}	   &${78}$  	\\\hline
			$169$&14&	$\le$86\cite{Griesmer1960}		&	${86}$	\\\hline
			$177$&22&	$\le$90\cite{Griesmer1960}	   &	${90}$	\\\hline
			$184$&29&	$\le$94\cite{Griesmer1960}		&	${94}$		\\\hline
			$200$&14&	$\le$102\cite{Griesmer1960}	    &	${102}$		\\\hline
			$208$&22&$\le$106\cite{Griesmer1960}		    &	${106}$		\\\hline
			$215$&29&	$\le$110\cite{Griesmer1960}		    &	${110}$		\\\hline
			$231$&14&	$\le$118\cite{Griesmer1960}  	    &${118}$			\\\hline
			$239$&22&$\le$122\cite{Griesmer1960}   &	${122}$			\\\hline
			$246$&29&$\le$126\cite{Griesmer1960}    &	${126}$
			
		\end{tabular}
	\end{center}
	\caption{New parameters of optimal $\dso(n,5)$ for $45\le n \le 256$ and $k=5$}\label{table:dso5}
\end{table}

\subsection{optimal self-orthogonal codes of dimension 6}
In this section, we introduce new parameters of optimal $[n,6]$ SO codes for lengths $n \le 256$. These SO codes are constructed using Algorithm \ref{alg1}.
There are examples of $[n,6]$ SO codes in \cite{Boukllieve2006}. In \cite{Boukllieve2006}, optimal $[n,6]$ SO codes of lengths only up to 40 are investigated and optimal $[n,6]$ SO codes of lengths $n \le 29$ and $n=32, 33, 34, 35, 36, 40$ are obtained as results. 

Besides optimal $[n,6]$ SO codes from  \cite{Boukllieve2006} and best-known linear codes($BKLC$) database of MAGMA\cite{Magma1994}, we successfully constructed many new optimal $[n,6]$ SO codes of lengths up to 256 starting from $BKLC$ database) of MAGMA\cite{Magma1994} and their punctured codes on at most six columns. The new parameters of SO codes of lengths up to 100 are listed in Table \ref{table:dso}, and all the examples of lengths up to 256 are presented in the web \cite{Choi2021}. In Table \ref{table:dso}, we denote the $\dso(n,6)$ by superscript $*$ if $\dso(n,6)=2\lfloor d(n,6)/2 \rfloor$ is confirmed by our construction result. 

Our computational result gives the following Theorem.

\begin{theorem} 
	For lengths $41 \le n \le 256$, if $n \nequiv 7, 14, 22, 29, 38, 45, 53, 60 \pmod{63}$ and $n \ne 46, 54, 61$, then $\dso(n,6) =  2\lfloor d(n,6)/2 \rfloor$.
\end{theorem}

\begin{table}
	
	\begin{center}
		
		\begin{tabular}{c|r|c||c|r|c}
			
			\multirow{2}{*}{$n$}	&\multicolumn{2}{c||}{$\dso(n,6)$}	&	\multirow{2}{*}{$n$}&	\multicolumn{2}{c}{$\dso(n,6)$}\\\cline{2-3}\cline{5-6}

			& Known[ref]	& Our result	&   &Known[ref]	&Our result\\\hline\hline
			$30$&$\le12~$\cite{Boukllieve2006}  &${12}$     &$141$	&$\le70~$\cite{Griesmer1960}    &	${70^*}$ 	    \\\hline
			$31$&$\le12~$\cite{Boukllieve2006}  &$\le12~$   &$142$	&$\le70~$\cite{Griesmer1960}	&	${70^*}$		\\\hline
			$37$&$\le16~$\cite{Boukllieve2006}  &${16^*}$   &$147$	&$\le72~$\cite{Griesmer1960}	&	${72^*}$		\\\hline
			$38$&$\le16~$\cite{Boukllieve2006}  & $16~$    	&$148$  &$\le74~$\cite{Griesmer1960}	&	${72}$			\\\hline
			$39$&$\le16~$\cite{Boukllieve2006}  &${16}$ 	&$149$	&$\le74~$\cite{Griesmer1960}	&	${74^*}$ 	    \\\hline
			$41$&$\le18~$\cite{Griesmer1960}    &${18^*}$   &$150$	&$\le74~$\cite{Griesmer1960}	&	${74^*}$	\\\hline
			$44$&$\le20~$\cite{Griesmer1960}	&${20^*}$  	&$154$	&$\le76~$\cite{Griesmer1960}	&	${76^*}$	\\\hline
			$45$&$\le22~$\cite{Griesmer1960}	&${20}$ 	&$155$	&$\le78~$\cite{Griesmer1960}	&	${76}$		\\\hline
			$46$&$\le22~$\cite{Griesmer1960}	&${20}$ 	&$156$  &$\le78~$\cite{Griesmer1960}	&	${78^*}$ 			\\\hline
			$47$&$\le22~$\cite{Griesmer1960}	&${22^*}$  	&$157$  &$\le78~$\cite{Griesmer1960}	&	${78^*}$ 			\\\hline
			$52$&$\le24~$\cite{Griesmer1960}	&${24^*}$   &$163$	&$\le80~$\cite{Griesmer1960}	&	${80^*}$	\\\hline
			$53$&$\le26~$\cite{Griesmer1960}	&${24}$ 	&$164$	&$\le82~$\cite{Griesmer1960}    &	${80}$	\\\hline
			$54$&$\le26~$\cite{Griesmer1960}	&${24}$ 	&$165$	&$\le82~$\cite{Griesmer1960}    &	${82^*}$	\\\hline
			$55$&$\le26~$\cite{Griesmer1960}	&${26^*}$	&$166$	&$\le82~$\cite{Griesmer1960}    &	${82^*}$		\\\hline
			$59$&$\le28~$\cite{Griesmer1960}	&${28^*}$	&$170$	&$\le84~$\cite{Griesmer1960}	&	${84^*}$		\\\hline
			$60$&$\le30~$\cite{Griesmer1960}    &${28}$     &$171$  &$\le86~$\cite{Griesmer1960}	&	${84}$		\\\hline
			$61$&$\le30~$\cite{Griesmer1960}	&${28}$     &$172$  &$\le86~$\cite{Griesmer1960}	&	${86^*}$		\\\hline
			$62$&$\le30~$\cite{Griesmer1960}    &${30^*}$	&$173$	&$\le86~$\cite{Griesmer1960}    &	${86^*}$	\\\hline
			$69$&$\le32~$\cite{Griesmer1960}	&${32^*}$ 	&$178$	&$\le88~$\cite{Griesmer1960}    &	${88^*}$		\\\hline
			$70$&$\le34~$\cite{Griesmer1960}    &${32}$     &$179$	&$\le90~$\cite{Griesmer1960}    &	${88}$ 	    \\\hline
			$71$&$\le34~$\cite{Griesmer1960}    &${34^*}$   &$180$	&$\le90~$\cite{Griesmer1960}	&	${90^*}$		\\\hline
			$72$&$\le34~$\cite{Griesmer1960}    &${34^*}$   &$181$	&$\le90~$\cite{Griesmer1960}	&	${90^*}$		\\\hline
			$76$&$\le36~$\cite{Griesmer1960}    &${36^*}$  	&$185$  &$\le92~$\cite{Griesmer1960}	&	${92^*}$			\\\hline
			$77$&$\le38~$\cite{Griesmer1960}    &${36}$ 	&$186$	&$\le94~$\cite{Griesmer1960}	&	${92}$ 	    \\\hline
			$78$&$\le38~$\cite{Griesmer1960}    &${38^*}$   &$187$	&$\le94~$\cite{Griesmer1960}	&	${94^*}$	\\\hline
			$79$&$\le38~$\cite{Griesmer1960}	&${38^*}$  	&$188$	&$\le94~$\cite{Griesmer1960}	&	${94^*}$	\\\hline
			$84$&$\le40~$\cite{Griesmer1960}	&${40^*}$ 	&$195$	&$\le96~$\cite{Griesmer1960}	&	${96^*}$		\\\hline
			$85$&$\le42~$\cite{Griesmer1960}	&${40}$ 	&$196$  &$\le98~$\cite{Griesmer1960}	&	${96}$ 			\\\hline
			$86$&$\le42~$\cite{Griesmer1960}	&${42^*}$  	&$197$  &$\le98~$\cite{Griesmer1960}	&	${98^*}$ 			\\\hline
			$87$&$\le42~$\cite{Griesmer1960}	&${42^*}$   &$198$	&$\le98~$\cite{Griesmer1960}	&	${98^*}$	\\\hline
			$89$&$\le44~$\cite{Griesmer1960}	&${44^*}$ 	&$202$	&$\le100~$\cite{Griesmer1960}    &	${100^*}$	\\\hline
			$90$&$\le44~$\cite{Griesmer1960}	&${44^*}$ 	&$203$	&$\le102~$\cite{Griesmer1960}    &	${100}$	\\\hline
			$91$&$\le44~$\cite{Griesmer1960}	&${44^*}$	&$204$	&$\le102~$\cite{Griesmer1960}    &	${102^*}$		\\\hline
			$92$&$\le46~$\cite{Griesmer1960}	&${44}$	    &$205$	&$\le102~$\cite{Griesmer1960}	&	${102^*}$		\\\hline
			$93$&$\le46~$\cite{Griesmer1960}    &${46^*}$   &$209$  &$\le104~$\cite{Griesmer1960}	&	${104^*}$		\\\hline
			$94$&$\le46~$\cite{Griesmer1960}	&${46^*}$   &$210$  &$\le104~$\cite{Griesmer1960}	&	${104^*}$		\\\hline
			$100$&$\le48~$\cite{Griesmer1960}   &${48^*}$	&$211$	&$\le106~$\cite{Griesmer1960}    &	${104}$	\\\hline
			$101$&$\le50~$\cite{Griesmer1960}	&${48}$ 	&$212$	&$\le106~$\cite{Griesmer1960}    &	${106^*}$		\\\hline
			$102$&$\le50~$\cite{Griesmer1960}   &${50^*}$   &$213$	&$\le106~$\cite{Griesmer1960}    &	${106^*}$ 	    \\\hline
			$103$&$\le50~$\cite{Griesmer1960}   &${50^*}$   &$217$	&$\le108~$\cite{Griesmer1960}	&	${108^*}$		\\\hline
			$107$&$\le52~$\cite{Griesmer1960}   &${52^*}$   &$218$	&$\le110~$\cite{Griesmer1960}	&	${108}$		\\\hline
			$108$&$\le54~$\cite{Griesmer1960}   &${52}$    	&$219$  &$\le110~$\cite{Griesmer1960}	&	${110^*}$			\\\hline
			$109$&$\le54~$\cite{Griesmer1960}   &${54^*}$ 	&$220$	&$\le110~$\cite{Griesmer1960}	&	${110^*}$ 	    \\\hline
			$110$&$\le54~$\cite{Griesmer1960}   &${54^*}$   &$226$	&$\le112~$\cite{Griesmer1960}	&	${112^*}$	\\\hline
			$115$&$\le56~$\cite{Griesmer1960}	&${56^*}$  	&$227$	&$\le114~$\cite{Griesmer1960}	&	${112}$	\\\hline
			$116$&$\le58~$\cite{Griesmer1960}	&${56}$ 	&$228$	&$\le114~$\cite{Griesmer1960}	&	${114^*}$		\\\hline
			$117$&$\le58~$\cite{Griesmer1960}	&${58^*}$ 	&$229$  &$\le114~$\cite{Griesmer1960}	&	${114^*}$ 			\\\hline
			$118$&$\le58~$\cite{Griesmer1960}	&${58^*}$  	&$233$  &$\le116~$\cite{Griesmer1960}	&	${116^*}$ 			\\\hline
			$122$&$\le60~$\cite{Griesmer1960}	&${60^*}$   &$234$	&$\le118~$\cite{Griesmer1960}	&	${116}$	\\\hline
			$123$&$\le62~$\cite{Griesmer1960}	&${60}$ 	&$235$	&$\le118~$\cite{Griesmer1960}    &	${118^*}$	\\\hline
			$124$&$\le62~$\cite{Griesmer1960}	&${62^*}$ 	&$236$	&$\le118~$\cite{Griesmer1960}    &	${118^*}$	\\\hline
			$125$&$\le62~$\cite{Griesmer1960}	&${62^*}$	&$241$	&$\le120~$\cite{Griesmer1960}    &	${120^*}$		\\\hline
			$132$&$\le64~$\cite{Griesmer1960}	&${64^*}$	&$242$	&$\le122~$\cite{Griesmer1960}	&	${120}$		\\\hline
			$133$&$\le66~$\cite{Griesmer1960}   &${64}$     &$243$  &$\le122~$\cite{Griesmer1960}	&	${122^*}$		\\\hline
			$134$&$\le66~$\cite{Griesmer1960}	&${66^*}$   &$244$  &$\le122~$\cite{Griesmer1960}	&	${122^*}$		\\\hline
			$135$&$\le66~$\cite{Griesmer1960}   &${66^*}$	&$248$	&$\le124~$\cite{Griesmer1960}    &	${124^*}$	\\\hline
			$139$&$\le68~$\cite{Griesmer1960}	&${68^*}$ 	&$249$	&$\le126~$\cite{Griesmer1960}    &	${124}$		\\\hline
			$140$&$\le70~$\cite{Griesmer1960}   &${68}$	    &$251$	&$\le126~$\cite{Griesmer1960}    &	${126^*}$	\\\hline

		\end{tabular}
		\caption{New parameters of $\dso(n,6)$ up to lengths $\le 256$}\label{table:dso}

	\end{center}
\end{table}

In \cite{Calderban1996}, a construction method of quantum codes from binary linear codes was introduced as following Theorem.
\begin{theorem}[Theorem 9 in \cite{Calderban1996}]\label{quantum}
	Let $\CC_1 \subseteq \CC_2$ be binary linear codes. By taking $\CC= \omega \CC_1 + \bar{\omega}\CC_2^\perp$, we obtain an $[[n, k_2-k_1,d]]$ code, where $$d=\min\{dist(\CC_2\backslash \CC_1),dist(\CC_1^\perp\backslash\CC_2^\perp) \}.$$
\end{theorem}

If we focus on binary self-orthogonal codes, we have the following corollary.
\begin{corollary}\label{quantum_cor}
	Let $\CC_1$ be a binary self-orthogonal $[n,k_1]$ code. If $d(\CC_1^\perp) \ge d(\CC_1)$, then by taking $\CC= \omega \CC_1 + \bar{\omega}\CC_1$, we obtain an $[[n, n-2k_1,d]]$ code, where $d=d(\CC_1^\perp).$
\end{corollary}
\begin{proof}
	Since $\CC_1$ is self-orthogonal, $\CC_1 \subseteq \CC_1^\perp$. Thus, letting $\CC_2 =  \CC_1^\perp$ in Theorem \ref{quantum}, we have $\CC_2^\perp=\CC_1$ and $k_2=n-k_1$. Therefore, $k_2-k_1=n-k_1$, and $\CC_2\backslash \CC_1=\CC_1^\perp\backslash \CC_1=\CC_1^\perp\backslash\CC_2^\perp $, thus  $d=dist(\CC_1^\perp\backslash \CC_1).$ By the assumption that $d(\CC_1^\perp) \ge d(\CC_1)$, $\CC_1^\perp$ has no non-zero codeword having weight smaller than $d(\CC_1)$. Thus,  $dist(\CC_1^\perp\backslash \CC_1)=d(\CC_1^\perp),$ and the corollary follows.
\end{proof}

\begin{example}\label{quantumcodeex}
	Using Algorithm \ref{alg1} with an optimal non-SO $[15,5,7]$ code, we construct a Reed-Muller $[16, 5, 8]$ SO code with generator matrix
	$$\left( \begin{smatrix}{cccccccccccccccc}
	1& 0& 0& 1 &0& 1 &1 &0& 0& 1 &1 &0 &1& 0& 0& 1\\
	0& 1& 0& 1 &0& 1 &0 &1& 0 &1 &0& 1& 0& 1& 0& 1\\
	0& 0&1 &1 &0 &0 &1 &1 &0 &0 &1 &1 &0 &0 &1 &1\\
	0& 0& 0& 0& 1& 1 &1 &1 &0& 0 &0 &0 &1& 1 &1 &1\\
	0& 0& 0& 0& 0& 0& 0& 0& 1& 1& 1 &1& 1 &1& 1 &1\\
	\end{smatrix}\right),$$
	whose dual code is a $[16,11,4]$ code. Thus, by Corollary \ref{quantum_cor}, we obtain an optimal $[[16,6,4]]$ quantum code by \cite{codetable}. In the same manner, we obtain $[[11,3,3]]$ and $[[15,7,3]]$ quantum codes which are optimal by \cite{codetable}.

\end{example}

\section{Concluding Remarks}
We have made three major contributions. First, we obtain a new method for checking self-orthogonality using the self-orthogonality matrix $SO_k$ and a vector $\ell(G)$. 
Second, we solve the problem of finding additional columns needed to make the shortest SO embedding code from a given binary $[n, k]$ code for any $k\ge 2$. Finally, we give the shortest SO embedding algorithm for the construction of optimal self-orthogonal codes. Using this algorithm, we succeed in obtaining many new optimal self-orthogonal codes of dimensions 5 and 6 for $n \le 256$.

As future work, it will be interesting to find new optimal SO codes with $n \ge 30$ and dimension $k \ge 7$.

%


\begin{thebibliography}{00}
	
	\bibitem{abbe2021} E. Abbe, et al., “Reed–Muller Codes: Theory and Algorithms,” IEEE Trans. Inform. Theory, vol. 67, pp. 3251-3277, 2021.
	
	\bibitem{Bachoc2004} C. Bachoc and P. Gaborit, “Designs and self-dual codes with long
	shadows,” J. Combin. Theory Ser. A, vol. 105, no. 1, pp. 15–34, 2004.
	
	
	\bibitem{Boukllieve2006} 	I. Bouyukliev, S. Bouyuklieva, T. A. Gulliver, and P. R. J. Ostergard, “Classification of optimal binary self-orthogonal codes,'' J. Comb. Math. Comb. Comput., vol. 59, p. 33, 2006.
	
	\bibitem{Calderban1996} A. R. Calderbank, E. M Rains, P. W. Shor, and N. J. A. Sloane, “Quantum Error Correction via Codes over GF(4),'' IEEE Trans. Inform. Theory, vol. 44, no. 4, pp, 1369-1387, 1996.
	
	
	\bibitem{Magma1994} J. Cannon, C. Playoust, {\it An Introduction to Magma.} University of Sydney, Sydney, Australia, 1994.
	
	\bibitem{Choi2021} W.-H. Choi, Code database in https://sites.google.com/view/whchoi, Accessed 2021-11-15.
	
	
	\bibitem{codetable} Markus Grassl,“Bounds on the minimum distance of linear codes and quantum codes,'' Online available at http://www.codetables.de. Accessed on 2021-10-22.
	
	
	\bibitem{Griesmer1960} J. H. Griesmer, “A bound for error-correcting codes,'' IBM J. Res. Develop., vol. 4, no. 5, pp. 532–542, 1960.
	
	\bibitem{Harada2009} M. Harada, “On the existence of frames of the Niemeier lattices and self-dual codes over Fp,” J. Algebra, vol. 321, no. 8, pp. 2345-2352, 2009.
	
	
	\bibitem{HP2010} W. C. Huffman, and V. Pless. {\it Fundamentals of Error-Correcting Codes.} Cambridge university press, 2010.
	
	
	\bibitem{Jin2010} L. Jin, S. Ling, J. Luo, and C. C. Xing, “Application of classical
	Hermitian self-orthogonal MDS codes to quantum MDS codes,” IEEE
	Trans. Inform. Theory, vol. 56, no. 9, pp. 4735–4740, 2010.
	
	\bibitem{Jin2012} L. Jin and C. Xing, “Euclidean and Hermitian self-orthogonal algebraic geometry codes and their application to quantum codes,” IEEE Trans. Inform. Theory, vol. 58, no. 8, pp. 5484–5489, 2012.
	
	
	\bibitem{Kim2002} J.-L. Kim, “New Quantum Error-Correcting Codes from Hermitian Self-Orthogonal Codes over GF(4),'' In: Mullen G.L., Stichtenoth H., Tapia-Recillas H. (eds) Finite Fields with Applications to Coding Theory, Cryptography and Related Areas. Springer, Berlin, Heidelberg, 2002. https://doi.org/10.1007/978-3-642-59435-9\_15  
	
	\bibitem{Kim2021} J. -L. Kim, Y. -H. Kim and N. Lee, “Embedding linear codes into self-orthogonal codes and their optimal minimum distances,'' IEEE Trans. Inform. Theory, vol. 67, no. 6, pp. 3701-3707, 2021, doi: 10.1109/TIT.2021.3066599.
	
	
	\bibitem{Lisoneks2014} P. Lisonek and V. Singh, “Quantum codes from nearly self-orthogonal quaternary linear codes,” Des. Codes Cryptogr., vol. 73, no. 2, pp. 417-424, 2014.
	
	\bibitem{Aileen1979} A. M. McLoughlin, “Covering radius of the $(m-3)$rd order Reed Muller codes and a lower bound on the $(m-4)$th order Reed Muller codes,'' SIAM J. Appl. Math., vol. 37, no. 2, pp. 419-422, 979.
	
	\bibitem{Niu2019} Y. Niu, Q. Yue, Y. Wu and L. Hu, “Hermitian Self-Dual, MDS, and Generalized Reed–Solomon Codes,'' IEEE Commun. Lett., vol. 23, no. 5, pp. 781-784, 2019, doi: 10.1109/LCOMM.2019.2908640.
	
	
	
	
	
\end{thebibliography}
\end{document}